\documentclass{article}
\usepackage{amsmath}
\usepackage{amsthm}
\usepackage{amssymb}
\usepackage{graphicx}
\usepackage{epstopdf, epsfig}
\usepackage{amsmath}
\usepackage{dashrule}
\usepackage{booktabs}
\usepackage{wasysym}
\RequirePackage[margin=3cm]{geometry}
\newtheorem{theorem}{Theorem}
\newtheorem{lemma}{Lemma}
\newtheorem{corollary}{Corollary}

\usepackage{multirow}
\usepackage{pbox}

\newcommand{\Rey}{Re}

\newcommand{\bs}[1]{\boldsymbol{#1}}

\usepackage{xcolor}
\usepackage{subcaption}

\title{Asymptotic scaling laws for periodic turbulent boundary layers and their numerical simulation up to $\Rey_{\theta}=8300$}

\author{Andrew Wynn, Saeed Parvar, Joseph O'Connor \& Sylvain Laizet}

\begin{document}

\maketitle

\begin{abstract}
We provide a rigorous analysis of the self-similar solution of the temporal turbulent boundary layer, recently proposed in \cite{Biau2023}, in which a body force is used to maintain a statistically steady turbulent boundary layer with periodic boundary conditions in the streamwise direction. We derive explicit expressions for the forcing amplitudes which can maintain such flows, and identify those which can hold either the displacement thickness or the momentum thickness equal to unity. This opens the door to the first main result of the paper, which is to prove upper bounds on skin friction for the temporal turbulent boundary layer. We use the Constantin-Doering-Hopf bounding method to show, rigorously, that the skin friction coefficient for periodic turbulent boundary layer flows is bounded above by a uniform constant which decreases asymptotically with Reynolds number. This asymptotic behaviour is within a logarithmic correction of well-known empirical scaling laws for skin friction. This gives the first evidence, applicable at asymptotically high Reynolds numbers, to suggest that the self-similar solution of the temporal turbulent boundary layer exhibits statistical similarities with canonical, spatially evolving, boundary layers. Furthermore, we show how the identified forcing formula implies an alternative, and simpler, numerical implementation of periodic boundary layer flows. We give a detailed numerical study of this scheme presenting direct numerical simulations up to $\Rey_\theta = 2000$ and implicit large-eddy simulations up to $\Rey_\theta = 8300$, and show that these results compare well with data from canonical spatially evolving boundary layers at equivalent Reynolds numbers.  
\end{abstract}

\section{Introduction}

Spatially evolving boundary layers are one of the canonical flows in fluid mechanics. Their behaviour governs the aerodynamic efficiency of aircraft, road vehicles, ships and even wind turbines via the behaviour of the atmospheric boundary layer. For this reason, there is significant interest in understanding and predicting key performance statistics, such as skin friction, of boundary layer flows. Resolved numerical simulations are a powerful tool with which to make such predictions. However, boundary layer flows of practical importance are typically at high Reynolds number, the boundary layer is itself turbulent, and the streamwise growth of its thickness necessitates a large computational domain in order to obtain converged and accurate time-averaged statistics. This places a severe restriction on the complexity of turbulent boundary layers that can be accurately simulated numerically. Indeed,  direct numerical simulations (DNS) of a turbulent boundary layer at the highest momentum Reynolds number to-date \cite{Sillero2013}, at $\Rey_\theta \approx 6500$, is at least an order of magnitude below that of the boundary layer on the fuselage of a commercial airliner flying at cruise velocity.

Owing to the computational demands of simulating spatially-developing turbulent flows, there has been a growing interest in temporal-based solutions. This involves, where possible, reframing a spatially-developing problem into a temporally-developing one, thus enabling a homogeneous solution in the streamwise direction. This simplifies the streamwise boundary conditions, which now become periodic, and also permits a shorter domain in the streamwise direction, thus reducing the computational cost. Such an approach has been adopted for a variety of flow problems, including mixing layers \cite{Rogers1994a} and planar jets \cite{VanReeuwijk2014a}. For the case of the spatially-developing boundary layer, the temporal reformulation becomes equivalent to the Rayleigh problem, where an infinitely long plate is impulsively started at constant velocity.

The first detailed analysis, using DNS, of the temporal turbulent boundary layer as a counterpart to the spatially-developing version was presented in \cite{Kozul2016}. This analysis showed that the temporal approach is a good model for the spatially-developing solution and is therefore a useful tool in the study of such flows. Furthermore, \cite{Kozul2016}  argued via an analysis of similarity solutions that the two approaches should become asymptotically equivalent at high Reynolds numbers. However, one of the challenges with their approach is that the final boundary layer thickness for a given Reynolds number is not known {\it a priori}. Therefore, both the domain and mesh are over-sized/resolved for the majority of the simulation, thus increasing the computational cost. Another disadvantage is the limited time window where statistics can be collected for a given Reynolds number. This imposes the requirement of having to run an ensemble of simulations to obtain converged statistics, which further increases computational demand. These issued have been solved \cite{Topalian2017a} by adapting the pioneering slow-growth formulation of \cite{Spalart1986a} from the original spatial-homogenisation approach to a temporal-homogenisation more suited to the temporally-developing turbulent boundary layer. However, the lack of a clearly defined temporal thickness growth rate introduced ambiguities with regards to extension towards general boundary layers with a non-zero pressure gradient. More recently, \cite{Biau2023} extended the work of \cite{Topalian2017a} by combining the temporal slow-growth formulation with the assumption of self-similarity. In addition to this, the solution was also non-dimensionalised with respect to the momentum thickness, so that the temporal thickness growth can be calculated to ensure the momentum thickness remains equal to unity throughout the simulation. This temporal-homogenisation reformulation is especially efficient from a numerical perspective, since, in addition to permitting a shorter streamwise domain via periodicity, it is statistically stationary in time and homogeneous in the streamwise and spanwise directions. This allows efficient mesh design for the duration of the simulation {\it a priori}, as well as accelerated statistical convergence through temporal averaging and spatial averaging in both the streamwise and spanwise directions.

The approach taken by \cite{Biau2023}, which we study in this paper, is simple in that it only involves adding a single forcing term of the form 
\[
f(t) y \frac{\partial \bs{u}}{\partial y} 
\]
to the governing Navier-Stokes equations, where $y$ is the wall-normal co-ordinate, $\bs{u}$ is the fluid velocity, and $f(t)$ is the amplitude of the forcing. One way to view this extra term is that it redirects streamwise momentum towards the boundary at $y=0$, adding energy to the flow. This is necessary if one wants to use periodic boundary conditions in the streamwise direction since, for a boundary layer geometry, the assumption of periodicity causes artificial energy dissipation to the extent that the resulting flow would not be representative of any finite streamwise section of a canonical, spatially evolving, boundary layer. To achieve a non-trivial statistically stationary flow the forcing amplitude $f(t)$ must therefore be chosen carefully. In the numerical scheme proposed by \cite{Biau2023}, at each time-step the forcing amplitude is defined implicitly via solution of an optimisation problem in order to control the value of an integral measure of the boundary layer thickness (either the displacement thickness, or momentum thickness) to be equal to a chosen fixed value, typically unity.

Given this method of creating `boundary-layer-like' flows on a periodic domains, it is of interest to determine the extent to which the resulting flow resembles a finite section of a canonical, spatially-evolving, boundary layer. In this paper, we make two contributions towards resolving this question, addressing both theoretical and numerical comparisons of periodic and spatially evolving boundary layer flows. 

From the theoretical perspective, we derive rigorous bounds on the turbulent statistics of periodic boundary layer flows. Our main result is to show that an upper bound on the time-averaged skin friction coefficient of the form $C_f \leq \upsilon(\Rey)$ holds, where $\Rey$ is a Reynolds number based on the boundary layer thickness. The upper bound $\upsilon(\Rey)$ is shown to be monotonically decreasing in $\Rey$ and converges to a constant 
\[
\lim_{\Rey \rightarrow \infty} \upsilon(\Rey) = \frac{1}{2\sqrt{2}}.
\]
 While the particular value of this constant is not important, we show that it is consistent (i.e. higher than) the skin friction values observed in numerical simulations. Of greater interest is that a uniform bound on skin friction is within a logarithmic correction of the well-known empirical frictions laws \cite[p.~577]{schlichting_2017} of the form
\[
C_f \sim \frac{1}{(\log{\text{Re}})^2}, 
\]
which are observed to hold for spatially evolving boundary layers for $\text{Re} \gg 1$. We show further that a logarithmic improvement to the uniform bound $C_f \leq \mathcal{O}(1)$ must hold if a weak assumption is made that the flow's turbulent kinetic energy grows logarithmically with $\text{\Rey}$. These are the first known bounds of this type proven to-date for boundary layer flows, with perhaps the most similar result in the literature being the bound on drag coefficient for a finite-length flat plate given by \cite{Kumar_Garaud_2020}. 

A key step towards proving the above theoretical bound is to alter Biau's numerical scheme for periodic turbulent boundary layers to remove the implicit definition of the forcing amplitude $f(t)$. We will show that an explicit formula for the forcing amplitude can be derived which enforces, asymptotically, a constant boundary layer thickness, and we show that this explicit formula agrees very accurately with the asymptotic forcing values arising from a numerical implementation of Biau's scheme. From the perspective of analysis, this gives a PDE which is amenable to applying the Contantin-Doering-Hopf upper bounding theory. However, from a numerical perspective it also opens the door to implementing an alternative numerical scheme to that presented in \cite{Biau2023}. 

In the second main contribution of the paper, we show how the introduced numerical approach can be used in both DNS and implicit large-eddy simulation (ILES) to simulate turbulent boundary layers on periodic domains, at a much lower cost than spatially evolving boundary layers. We present a detailed analysis of the turbulent statistics of these flows for DNS up to $\Rey_\theta =2000$ and ILES up to $\Rey_\theta = 8300$, first validating the new numerical scheme against the results of \cite{Biau2023} for periodic turbulent boundary layers, and then comparing our results with the existing DNS of \cite{Sillero2013} and \cite{Schlatter2010}, and with the ILES of \cite{Eitel-Amor2014}, for canonical, spatially evolving, boundary layers.

\subsection{Problem Setup}

Suppose that a fluid with velocity $\bs{u} = u\bs{e}_x + v \bs{e}_y + w \bs{e}_z$ occupies a semi-infinite rectangular domain $\Omega = [0,L_x] \times [0,\infty) \times [0,L_z] \subset \mathbb{R}^3$, and is confined by an impermeable wall at $y=0$ where no-slip boundary conditions are imposed. The flow is driven by a free-stream velocity $U_\infty \bs{e}_x$ infinitely far from the wall, and periodic boundary conditions are assumed in the streamwise $\bs{e}_x$ and spanwise  $\bs{e}_{z}$ directions. Using $U_\infty$ as a velocity scale and unity as a length-scale, the non-dimensional Navier-Stokes equations for the flow are
\begin{equation} \label{eq:NSE}
\begin{split}
\frac{\partial \bs{u}}{\partial t}+(\bs{u} \cdot \nabla )\bs{u} + \nabla p  & = \frac{1}{\Rey}\Delta \bs{u} + \bs{f}  \\
\nabla \cdot \bs{u} &=0 
\end{split}
\end{equation}
and, following \cite{Biau2023}, use the boundary conditions
\begin{equation} \label{eq:bcs}
\begin{split}
\lim_{y \rightarrow \infty} \bs{u}(x,y,z,t) &= 1\cdot \bs{e}_x,\\
\bs{u}(x,0,z,t)&=\bs{0},\\
\bs{u} \sim \text{periodic in the} &\; \bs{e}_x \, \text{and} \, \bs{e}_z \; \text{directions}. 
\end{split}
\end{equation}
Here $\Rey = \frac{U_\infty}{\nu}$ is the Reynolds number, $\bs{f}$ is a non-dimensional body force, and $\nu$ is the kinematic viscosity.  The unusual choice of using unity as a length scale is motivated by the fact that the body force will be subsequently chosen to control the boundary layer thickness to unity, thus making this an appropriate spatial length scale for the problem. 

To describe the boundary layer thickness on the periodic domain $\Omega$, let the streamwise and spanwise averaged $\bs{e}_x$ velocity be given by  
\[
U(y,t):=\langle u \rangle :=\frac{1}{L_x L_z} \int_0^{L_x} \int_0^{L_z} u(x,y,z,t) dx dz, \qquad y,t \geq 0,
\]
and define the instantaneous displacement thickness, $\delta^{*}(t)$, and momentum thickness, $\theta(t)$, by
\[
\delta^{*}(t) = \int_0^\infty (1- U(y,t) ) \, dy \qquad \text{and} \qquad \theta(t) := \int_0^\infty U(y,t)(1-U(y,t)) \,dy.
\]
In order to control either $\delta^\ast$ or $\theta$ to be unity, throughout this paper we use the approach of \cite{Biau2023} and let the body force  be given by 
\begin{equation} \label{eq:body_force}
\bs{f}(t) = f(t) y \frac{\partial \bs{u}}{\partial y},
\end{equation}
where $f(t) \in \mathbb{R}$ is the forcing amplitude.

\section{Boundary layer thickness control} \label{sec:control}

The aim of this section is to derive explicit expressions for the forcing amplitude $f(t)$ in \eqref{eq:body_force} under which either of the boundary layer thicknesses $\delta^\ast(t)\equiv 1$ or $\theta(t) \equiv 1$ can be maintained for solutions to \eqref{eq:NSE} which satisfy the periodic boundary conditions \eqref{eq:bcs}.  Figure \ref{fig:bl_control_blocks} shows a schematic overview of how such a forcing $f(t)$ is defined. It consists of two components,
\[
f = K( e) + F( U , \langle uv \rangle),
\]
where $K$ and $F$ are nonlinear functionals, both of which will be defined subsequently. 

The first term, $K$, is an error feedback controller that uses the deviation $e=e_\delta=1-\delta^\ast$ (or $e = e_\theta = 1- \theta$) from the desired unity value of boundary layer thickness. The second term, $F$, which depends only on the mean streamwise velocity $U$ and the stresses $\langle uv\rangle$, is a feed-forward control term that will be constructed from analysis of the governing equations. Importantly, this will allow us to identify an explicit formula for the forcing $f$ once the flow is in a statistically steady state.   

\begin{figure}
\centering
\includegraphics[width=0.75\linewidth, trim={0.cm 0cm 0.cm 0.0cm},clip]{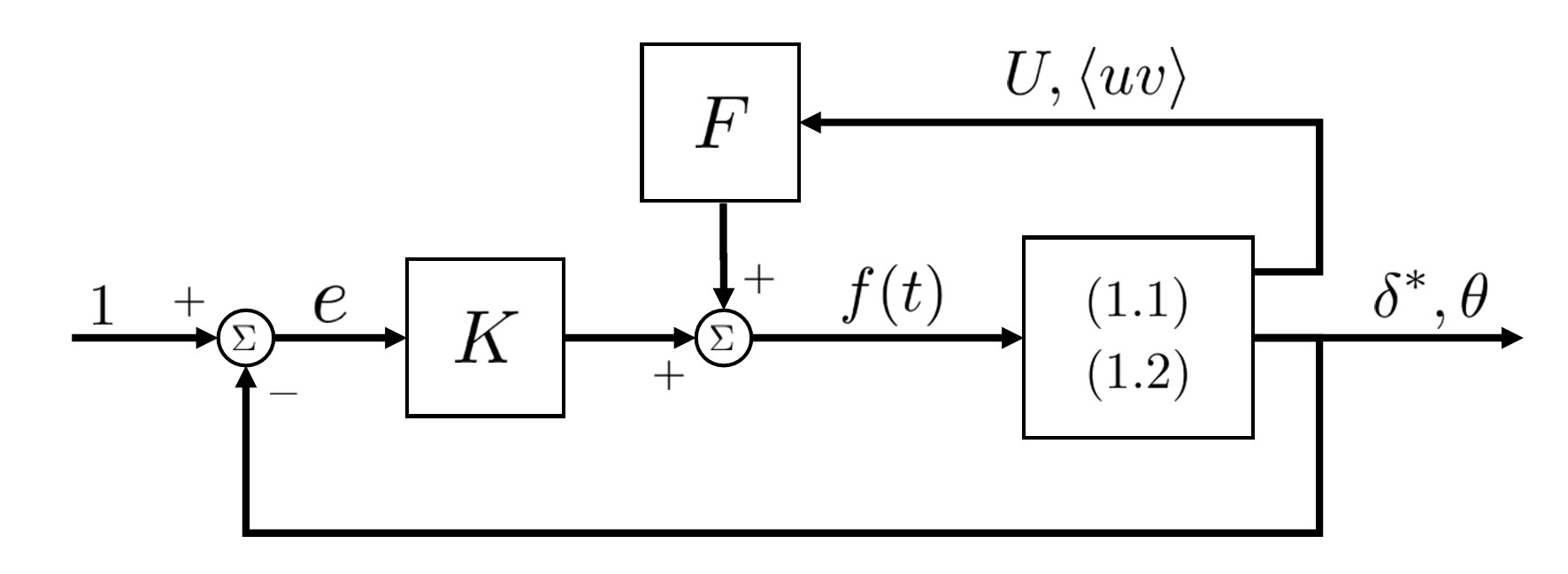}
\caption{A schematic overview of the proposed boundary layer thickness control scheme. The idea is to force either $\delta^\ast$ or $\theta$ to converge to a reference value of $1$, by choice of the nonlinear control laws $K$ and $F$.  }
\label{fig:bl_control_blocks}
\end{figure}

\subsection{Displacement thickness control}

To determine appropriate choices of the controllers $K$ and $F$, first let $e_\delta(t) = 1-\delta^\ast(t)$. It is shown in the Appendix that
\begin{equation} \label{eq:displacement_evolution}
\frac{d e_\delta}{dt} = f(t)\delta^\ast(t) - \frac{1}{2} C_f(t),
\end{equation}
where we define the instantaneous streamwise- and spanwise-averaged skin friction by 
\[
C_f(t) := \frac{2}{\Rey} \frac{\partial U}{\partial y}(0,t), \qquad t \geq 0.
\]

Given \eqref{eq:displacement_evolution}, it is not difficult to see that asymptotic decay, i.e. $e_\delta(t) \rightarrow 0$, of the error can be ensured by defining the feedback and feed-forward components of the forcing amplitude $f(t)$ by
\[
K(e) = -\frac{ke}{1-e} = - \frac{ke_\delta}{\delta^\ast}, 
\]
and
\[
F(U) = \frac{\frac{ \partial U}{\partial y}(0,t)}{\Rey \int_0^\infty (1-U) dy} = \frac{C_f}{2\delta^\ast},
\]
respectively, where $k>0$ is a constant gain parameter. This observation gives the following result.

\begin{lemma} \label{lem:displacement_control}
Let $e_\delta = 1-\delta^\ast$. Suppose that the forcing amplitude \eqref{eq:body_force} is chosen such that 
\begin{equation} \label{eq:displacement_forcing}
f(t) := \frac{-k e_\delta(t) + \frac12  C_f(t)}{\delta^{*}(t)},
\end{equation}
for some $k >0$. Then, for any solution to \eqref{eq:NSE} with boundary conditions \eqref{eq:bcs} the displacement thickness satisfies $\delta^{*}(t) \rightarrow 1$ as $t \rightarrow \infty$.
\end{lemma}

An important corollary of Lemma \ref{lem:displacement_control} is that 
\begin{equation} \label{eq:displacement_f_limit}
\lim_{t\rightarrow \infty} \left|f(t) - \frac12 C_f(t) \right| =0,
\end{equation}
which implies that once the flow has reached a statistically-steady state it must satisfy the PDE 
\begin{equation} \label{eq:NSE_d1}
\begin{split}
\frac{\partial \bs{u}}{\partial t}+(\bs{u} \cdot \nabla )\bs{u} +  \nabla p  & = \frac{1}{\Rey} \left( \Delta \bs{u} + y \frac{\partial \bs{u}}{\partial y}\frac{\partial U}{\partial y}(0,t) \right),  \\
\nabla \cdot \bs{u} &=0,
\end{split}
\end{equation}
subject to the boundary conditions \eqref{eq:bcs}. 

The fact that \eqref{eq:NSE_d1} is an explicit form of the governing equations for a periodic boundary layer flow is of significant use for analysis. We will exploit this later in \S \ref{sec:bounds} to derive rigorous estimates of the skin friction for solutions to \eqref{eq:NSE_d1}. In contrast, the numerical scheme first proposed in \cite{Biau2023} defines the forcing $f(t)$ {\em implicitly}, via the solution to an optimisation problem, and this  makes a careful analysis of the governing equations more challenging. Nonetheless, we will show in \S \ref{sec:numerical_results} that a numerical implementation of \eqref{eq:NSE} with $f(t)$ defined by either \eqref{eq:displacement_forcing}, or by using Biau's implicit method, give rise to the same flows.

\subsection{Momentum thickness control} \label{sec:control_theta}

Momentum thickness can be controlled in an analogous manner to the displacement thickness. Letting $e_\theta(t) = 1-\theta(t)$, it is shown in the Appendix that 
\begin{equation} \label{eq:momentum_ode}
\frac{de_\theta}{dt} = f(t) \theta(t) + \frac12 C_f(t) - \frac{2}{\Rey} \int_0^\infty \left( \frac{ \partial U}{\partial y}\right)^2 dy + 2 \int_0^\infty \langle uv \rangle \frac{\partial U}{\partial y} \, dy.
\end{equation}
In this case,  the error feedback is defined by 
\[
K(e_\theta) = \frac{-ke_\theta}{1-e_\theta} = \frac{-ke_\theta}{\theta}
\]
and the feed-forward term given by
\[
F(U,\langle uv \rangle) = \frac{ \frac{1}{\Rey} \frac{\partial U}{\partial y}(0) - \frac{2}{\Rey} \int_0^\infty \left( \frac{ \partial U}{\partial y}\right)^2 dy + 2 \int_0^\infty \langle uv \rangle \frac{\partial U}{\partial y} \, dy}{ \int_0^\infty U(1- U) dy}
\]
Letting $f(t) = K(e) + F(U,\langle uv \rangle)$ and substituting into \eqref{eq:momentum_ode} then gives the following result.

\begin{lemma} \label{thm:momentum_control}
Let $e_\theta(t) = 1-\theta(t)$. Suppose that the forcing amplitude \eqref{eq:body_force} is chosen such that 
\begin{equation} \label{eq:momentum_control}
f(t) := \frac{-ke_\theta(t)   - \frac12 C_f(t) + \frac{2}{\Rey} \int_0^\infty \left( \frac{ \partial U}{\partial y}\right)^2 dy - 2 \int_0^\infty \langle uv \rangle \frac{\partial U}{\partial y} \, dy}{\theta(t)},
\end{equation}
for some constant $k>0$. Then for any solution to \eqref{eq:NSE} satisfying the boundary conditions \eqref{eq:bcs}, the momentum thickness satisfies $\theta(t) \rightarrow 1$ as $t \rightarrow \infty$.
\end{lemma}

A consequence of Lemma \ref{thm:momentum_control} is that an explicit formula can be found for the asymptotic forcing 
\begin{equation} \label{eq:f_asymptotic_momentum}
 \lim_{t \rightarrow \infty} \left| f(t) - \left[ \frac12 C_f(t) - 2 \int_0^\infty \frac{1}{\Rey} \left( \frac{\partial U}{\partial y} \right)^2 - \frac{\partial U}{\partial y} \langle uv \rangle \, dy \right] \right| =0.
\end{equation}
This result will be supported by numerical evidence presented in \S \ref{sec:numerical_results}, that confirms the asymptotic limit if $f(t)$ is given either by \eqref{eq:momentum_control}, or computed implicitly using the method of \cite{Biau2023}.

\section{Rigorous bounds on skin friction} \label{sec:bounds}

We now address the question of whether solutions to \eqref{eq:NSE}, \eqref{eq:bcs} on the periodic domain $\Omega$ have comparable statistical properties to those of a canonical, spatially evolving, boundary layer. In particular, our aim is to prove rigorous upper bounds on the time-averaged value of skin friction, 
\[
C_f := \limsup_{t \rightarrow \infty} \frac{1}{T} \int_0^T C_f(t) dt,
\]
and to understand how these bounds scale with $\Rey$. It is important to note that even though the results of \S \ref{sec:control} ensure that the boundary layer thickness is kept constant, there is no {\it a priori} guarantee that the underlying fluid velocity field itself, or its statistics such as $C_f$, are well behaved. 

In this section, we consider solutions to \eqref{eq:NSE}, \eqref{eq:bcs} with the forcing $f(t)$ given by \eqref{eq:displacement_forcing}. This ensures that the displacement thickness satisfies $\delta^\ast(t) \rightarrow 1$ as $t \rightarrow \infty$. We therefore assume without loss of generality that \eqref{eq:NSE_d1} are the governing equations for the flow and that $\delta^\ast = 1$. Consequently, the Reynolds number should be interpreted as 
\[
\Rey = \Rey_{\delta^\ast} = \frac{U_\infty \delta^\ast}{\nu} = \frac{U_\infty}{\nu}. 
\]

Throughout this section,  we also only consider solutions to \eqref{eq:NSE_d1} that satisfy both $C_f(t) \geq 0, t \geq 0$, and the decay condition 
\begin{equation} \label{eq:decay}
\lim_{y \rightarrow \infty} y(\bs{e}_x-\bs{u}(x,y,z)) = 0, \qquad  0 \leq x \leq L_x, \, 0 \leq z \leq L_z, 
\end{equation}
which places only a weak constraint on the following analysis. A similar assumption was made in \cite{Kumar_Garaud_2020} when bounding the drag coefficient for flow past a finite-length flat plate.

We begin by taking the dot product of \eqref{eq:NSE_d1} with $\bs{u}$, integrating by parts, using the boundary conditions \eqref{eq:bcs}, the assumption \eqref{eq:decay}, and the fact that $\delta^\ast=1$ to obtain the energy equation 
\begin{equation} \label{eq:displacement_control_E}
\frac 12 \frac{d}{dt}  \| \bs{u} \|^2  + \frac{1}{\text{\Rey}}\| \nabla \bs{u} \|^2 = \frac12 C_f(t) \left(1 - \frac12 \|\bs{u}-\bs{e}_x\|_2^2 \right),
\end{equation}
a proof of which is given in the Appendix. For clarity, in the above equation, we have used the notation 
\[
\|\bs{u}\|^2 = \int_\Omega \bs{u} \cdot \bs{u} \, d\bs{x}.
\]

Next, consider a decomposition of the velocity field 
\begin{equation} \label{eq:BF_decomposition}
\bs{u}(\bs{x},t) = (1-\psi(y))\bs{e}_x + \bs{v}(\bs{x},t)
\end{equation}
where $\psi :[0,\infty) \rightarrow \mathbb{R}$ is a function satisfying 
\[
\psi(0)=1 \quad \text{and} \quad \lim_{y \rightarrow \infty} y\psi(y) = 0. 
\]
The perturbation $\bs{v} = v_x \bs{e}_x + v_y \bs{e}_y + v_z \bs{e}_z$ then inherits the periodic boundary conditions in the $\bs{e}_x,\bs{e}_z$ directions and also satisfies
\begin{equation} \label{eq:v_bcs}
\bs{v}(x,0,z,t)=0 \quad \text{and} \quad \lim_{y \rightarrow \infty} y\bs{v}(x,y,z,t)=0,
\end{equation}
for any fixed $x,z$ and $t$. Using these properties, the following energy equation can be derived for the perturbations $\bs{v}$, 
\begin{align} 
\frac12 \frac{d}{dt}\|\bs{v}\|^2 &+ \frac{1}{\Rey} \|\nabla \bs{v}\|^2 - \int v_x v_y \psi' d\bs{x}\nonumber \\
\qquad \qquad &= \frac{1}{\Rey}\int \frac{\partial v_x}{\partial y} \psi' d\bs{x} + \frac12 C_f(t) \int y \bs{v} \cdot \frac{\partial \bs{u}}{\partial y} d\bs{x}. \label{eq:energy_v}
\end{align}
To simplify this expression, we make use of the identities
\[
\|\nabla \bs{u}\|^2 = \|\nabla \bs{v}\|^2 + \|\psi'\|^2 - 2 \int\frac{\partial v_x}{\partial y} \psi' d\bs{x}.
\]
and
\[
\int y \bs{v} \cdot \frac{\partial \bs{u}}{\partial y} d\bs{x} = -\frac12 \|\bs{v}\|^2 - \int y v_x \psi' d\bs{x}
\]
which, when substituted into \eqref{eq:energy_v}, give
\[
\begin{split}
&\frac12 \frac{d}{dt}\|\bs{v}\|^2 + \frac{1}{2\Rey}\left( \| \nabla \bs{v} \|^2 + \|\nabla \bs{u}\|^2 \right) - \int v_x v_y \psi' d\bs{x}\\
&\qquad \qquad +  \frac12 C_f(t) \left( \frac12 \|\bs{v}\|^2 + \int y v_x \psi' d\bs{x}\right) = \frac{1}{2\Rey} \|\psi'\|^2. 
\end{split}
\]
Now, defining $Q_\psi(\bs{v}):=(\Rey^{-1}) \|\nabla \bs{v}\|^2 - 2\int v_x v_y \psi' dy$ and using \eqref{eq:displacement_control_E} gives
\begin{equation} \label{eq:dv_du}
\begin{split}
&\frac{d}{dt}\left[ 2\|\bs{v}\|^2 - \|\bs{u}\|^2 \right] + 2 Q_\psi (\bs{v}) \\
&\qquad + C_f(t) \left(1 - \frac12\|\bs{u}-\bs{e}_x\|^2 + \|\bs{v}\|^2 + 2 \int y v_x \psi' dy \right) = \frac{2}{\Rey} \|\psi'\|^2
\end{split}
\end{equation}

The terms multiplying $C_f(t)$ can be bounded from below via 
\begin{align}
1 - \frac12\|\bs{u}-\bs{e}_x\|^2 &+ \|\bs{v}\|^2 + 2 \int y v_x \psi' \, d\bs{x} \nonumber \\
&= 1 + \frac12 \|\bs{v}\|^2 - \frac12 \|\psi\|^2 + \int v_x(\psi + 2y \psi') d\bs{x} \nonumber \\
&\geq 1 +  \frac12\|\bs{v}\|^2 - \frac12 \|\psi\|^2 - \frac12 \|v_x\|^2 - \frac12 \|\psi + 2y\psi'\|^2 \nonumber \\
&=1 + \frac12(\|v_y\|^2 + \|v_z\|^2 ) - 2 \|y\psi'\|^2 \label{eq:cf_multipliers2}\\
& \geq 1 - 2\|y\psi'\|^2. \label{eq:cf_multipliers}
\end{align}
Under the assumption that $C_f(t) \geq 0$, it then follows from \eqref{eq:displacement_control_E} that both $\|\bs{u}\|^2$ and $\|\bs{v}\|^2$ are uniformly bounded in time. Consequently, after time-averaging \eqref{eq:dv_du} to remove the time derivatives, using $C_f \geq 0$ again, and the estimate \eqref{eq:cf_multipliers}, we have the upper bound 
\[
2\overline{Q_\psi(\bs{v})} + C_f(1-2\|y\psi'\|^2) \leq \frac{2}{\Rey}\|\psi'\|^2.
\]

Now, suppose that $\psi'$ can be chosen so that the {\em spectral constraint} 
\[
Q_\psi(\bs{v}) \geq 0
\]
holds for any $\bs{v}$ satisfying the boundary conditions \eqref{eq:v_bcs}. Since any solution to the PDE \eqref{eq:NSE_d1} gives rise to a perturbation satisfying these conditions, it would then follow that the skin friction was upper bounded by 
\begin{equation} \label{eq:bound}
C_f \leq \frac{2\|\psi'\|^2}{\Rey(1 - 2 \|y\psi'\|^2)}. 
\end{equation} 
This observation leads to the main result of the paper. 

\begin{theorem} \label{thm:bounds}
Let $\Rey > 1/\sqrt{2}$. For any solution of \eqref{eq:NSE_d1} satisfying the boundary conditions \eqref{eq:bcs} for which $C_f(t) \geq 0$ and \eqref{eq:decay} hold, the time-averaged skin friction satisfies
\[
C_f \leq \frac{\Rey}{2(\sqrt{2}\Rey -1)}.
\]
\end{theorem}
\begin{proof}
Let $\psi(y)$ be defined by 
\[
\psi(y) = e^{-\frac{\Rey}{\sqrt{2}}y}, \qquad y \geq 0.
\]
Then $\psi(0)=1$ and $y\psi(y)\rightarrow 0$ as $y \rightarrow \infty$. Now, standard estimates imply that  
\[
\left| \int v_x v_y \psi' d\bs{x}\right| \leq \frac{1}{2\sqrt{2}} \|\nabla \bs{v} \|^2 \int y |\psi'(y)| dy = \frac{1}{2\Rey} \|\nabla \bs{v}\|^2, 
\]
for any $\bs{v}$ satisfying $\bs{v}(x,0,z)=0$. Consequently, the spectral constraint $Q_\psi \geq 0$ is satisfied. It then follows from \eqref{eq:bound} that
\[
C_f \leq  \frac{2\|\psi'\|^2}{\Rey(1 - 2 \|y\psi'\|^2)} = \frac{\Rey / (2\sqrt{2})}{1 -1 / (\sqrt{2}\Rey)} = \frac{\Rey}{2(\sqrt{2}\Rey -1)}.
\] 
\end{proof}

An interesting corollary of Theorem \ref{thm:bounds} can be obtained by retaining the perturbation energy term $E_{\bs{v}}(t):= \frac12 \left( \|v_y\|^2 + \|v_z\|^2\right)$ in the estimate \eqref{eq:cf_multipliers2}. In Theorem \ref{thm:bounds} these terms are estimated from below by zero, but one would expect that their time average satisfies $\overline{E_v} >0$ for any turbulent solution to the governing equations. To make use of this observation, we introduce the notion of the covariance of two time-series $a(t)$ and $b(t)$, letting 
\[
\text{cov}(a,b) := \lim_{T\rightarrow \infty} \frac{1}{T} \int_0^T ( a(t) - \bar{a})(b(t)-\bar{b}) dt. 
\]
It follows from this definition that the time average of the product $a(t)b(t)$ can be expressed as  
\[
\overline{a b} = \bar{a}\bar{b} + \text{cov}(a,b). 
\]
Applying this formula when time-averaging \eqref{eq:dv_du} then gives
\[
\overline{Q_\psi(\bs{v})} + C_f (1 + \overline{E_{\bs{v}}} -2\|y\psi'\|^2) \leq \frac{2\|\psi'\|^2}{\Rey} - \text{cov}(C_f,E_{\bs{v}})
\]
We then have the following corollary, whose proof uses the same function $\psi$ as in Theorem \ref{thm:bounds}.
\begin{corollary} \label{corr:log_correction}
For any solution of \eqref{eq:NSE_d1} for which $C_f(t) \geq 0$ and \eqref{eq:decay} hold,  the time-averaged skin friction satisfies
\[
C_f \leq \frac{1}{2\sqrt{2}(1 + \overline{E_{\bs{v}}})} - \frac{\mathrm{cov}(C_f,E_{\bs{v}})}{1 + \overline{E_{\bs{v}}}}, \qquad \Rey \gg 1. 
\]
\end{corollary}

\subsection{Discussion of the theoretical results}

One should not expect to match the value of the constant in the analytically provable bound $C_f \leq (2\sqrt{2})^{-1}$ with the true value of $C_f$ observed in numerical simulations, although numerically optimising $\psi$ in Constantin-Doering-Hopf bounding arguments can often improve the absolute value of such constants by an order of magnitude (see, for example, studies by \cite{PLASTING_KERSWELL_2003} and \cite{Fantuzzi16}). However, one may instead hope to capture the correct asymptotic scaling of $C_f$ with Reynolds number. From this perspective, the result Theorem \ref{thm:bounds} is  more important. That  $C_f$ is provably bounded, independent of the Reynolds number, is within a logarithmic correction of empirically-observed scaling expected for spatially evolving boundary layers.

Another view is provided if we define a friction Reynolds number based on the displacement thickness by
\[
\Rey_\tau = \frac{U_\tau \delta^{*}}{\nu} = Re_{\delta^\ast} \sqrt{ \frac{C_f}{2}}.
\]
Theorem \ref{thm:bounds} then implies that 
\[
\Rey_\tau \leq 2^{-\frac54} \Rey_{\delta^{\ast}}.
\] 
This is not dissimilar to the observation of \cite{Schlatter2010} that for canonical, spatially-evolving, boundary layers 
\[
\text{Re}_\tau \approx c \Rey_{\delta^\ast}^{0.843}, 
\]
where  $c = 1.13 (\delta^{\ast})^{0.157} \theta^{0.843} \delta^{-1} = \mathcal{O}(1)$ is simply a constant accounting for the various possible definitions of the boundary layer thickness, and $\delta$ is the boundary layer thickness defined in terms of $99\%$ of the freestream velocity.  

The exponent $0.843$ in this observation should be treated with caution, being based on an empirical fit to data in a small (in the context of understanding asymptotic scalings) range of Reynolds numbers. Indeed, the high Reynolds number scaling corresponding to the logarithmic friction law would be of the form 
\[
\Rey_\tau \sim \frac{\Rey_{\delta^\ast}}{\log{\Rey_{\delta^\ast}}}, \qquad \Rey_{\delta^\ast} \gg 1. 
\]

Whether an analytical proof can bridge this `logarithmic' gap is an important open question in theoretical fluid mechanics. It is therefore of interest that Corollary \ref{corr:log_correction} provides partial evidence that the origin of such a correction can now be understood, at least for periodic boundary layers. Indeed, the improved bound 
\[
C_f \leq \frac{1}{2\sqrt{2}(1 + \overline{E_{\bs{v}}})} - \frac{\mathrm{cov}(C_f,E_{\bs{v}})}{1 + \overline{E_{\bs{v}}}}
\]
would imply such a logarithmic correction, if it could be proven that the 
\[
\sup_{\Rey_{\delta^\ast} >0} |\mathrm{cov}(C_f,E_{\bs{v}})| < \infty \quad \text{and} \quad \lim_{\Rey_{\delta^\ast} \rightarrow \infty} \left( \frac{\overline{E_{\bs{v}}}}{\log{\Rey_{\delta^\ast}}} \right)> 0,  
\]
that is, if the turbulent perturbations themselves exhibit a logarithmic growth of energy and are only weakly correlated with the instantaneous skin-friction.

\section{Numerical Implementation}

To perform DNS and ILES of the periodic boundary layer equations \eqref{eq:NSE} we use the \texttt{Xcompact3d} framework, a suite of fluid flow solvers dedicated to the study of turbulent flows on supercomputers \cite{Bartholomew2020}, which has been extensively validated for wall-bounded turbulent flows \cite{diaz2017wall,Mahfoze2021,o2023optimisation}. The ILES are based on a strategy that introduces targeted numerical dissipation at the small scales through the discretisation of the second derivatives of the viscous terms \cite{Lamballais2011,Dairay2017}. It was shown in these studies that it is possible to design a high-order
finite-difference scheme in order to mimic a subgrid-scale model for ILES based on the concept of Spectral Vanishing Viscosity, at no extra computational cost and with excellent performance for wall-bounded turbulent flows \cite{Mahfoze2021}.

The incompressible flow solver within \texttt{Xcompact3d} is based on sixth-order compact finite-difference schemes \cite{Laizet2009} for the spatial discretisation and a fractional-step method using a semi-implicit approach that combines the Crank-Nicholson and 3rd-order Adams–Bashforth methods. Within the fractional-step method, the incompressibility condition is dealt with by directly solving a Poisson equation in spectral space using 3D Fast Fourier Transforms (FFTs) and the concept of the modified wavenumber \cite{Lele1992}. The velocity-pressure mesh arrangement is half-staggered to avoid spurious pressure oscillations \cite{Laizet2009}.

The simplicity of the mesh allows an easy implementation of a 2D domain decomposition based on pencils \cite{laizet2011incompact3d}. The computational domain is split into a number of sub-domains (pencils) which are each assigned to an MPI-process. The derivatives and interpolations in the x-direction (y-direction, z-direction) are performed in X-pencils (Y-pencils, Z-pencils), respectively. The 3D FFTs required by the Poisson solver are also broken down as a series of 1D FFTs computed in one direction at a time. Global transpositions to switch from one pencil to another are performed with the MPI command \verb|MPI_ALLTOALL(V)|. The flow solvers within \texttt{Xcompact3d} can scale well with up to hundreds of thousands of MPI-processes for simulations with several billion mesh nodes \cite{laizet2011incompact3d,Bartholomew2020}. All the simulations in this study were carried out on ARCHER2, the UK supercomputing facility. It is equipped with nodes based on dual AMD EPYC$^{TM}$ 7742 processors running at 2.25 GHz, totaling 128 cores per node.

In this numerical study, the following incompressible Navier-Stokes equations are solved 
\begin{equation} \label{eq:ILES_1}
\begin{split}
\frac{\partial \bs{u}}{\partial t}+ \frac{1}{2} [\nabla \cdot (\bs{u} \otimes\bs{u})+(\bs{u} \cdot \nabla)\bs{u}]&=-\nabla p + (I-\bs{\Gamma})\mathcal{D}_1+\bs{\Gamma}\mathcal{D}_2 +\bs{f}(t),  \\
\nabla \cdot \bs{u}&=0.
\end{split}
\end{equation}

With a slight abuse of notation, $\mathcal{D}_1$ indicates that the diffusion term $\nu \Delta \bs{u}$ is directly implemented using a conventional sixth-order finite-difference scheme for its second-order spatial derivatives; while $\mathcal{D}_2$ refers to a numerical implementation of $\nu \Delta \bs{u}$ which adds targeted numerical dissipation at small scales via a customised sixth-order finite-difference scheme.  Full details can be found in 
\cite{Mahfoze2021}, \cite{Dairay2017} and \cite{Lamballais2011}. The operator $\bs{\Gamma}$ is used to balance the weight of the two diffusive terms and can depend on the both local geometry and flow velocities via $(\bs{\Gamma}\bs{u})(x) = \Gamma(x,\bs{u})\bs{u}(x)$, where $\Gamma(x,\bs{u}) \in \mathbb{R}^{3 \times 3}$.  The choice $\Gamma=0$ corresponds to a DNS implementation, in which \eqref{eq:ILES_1} is mathematically equivalent to \eqref{eq:NSE}. The case $\Gamma = I$ corresponds to ILES, for which, the unknowns $\mathbf{u}(x,t)$ and $p(x,t)$ should be interpreted as the large-scale component of velocity and pressure. Note finally that the advection terms in \eqref{eq:ILES_1} are written in skew-symmetric form in order to reduce aliasing errors \cite{Kravchenko1997}. 

%

\subsection{Numerical Implementation of the forcing term}

To allow for a validation with the numerical study of \cite{Biau2023}, we will perform simulations in the case of momentum thickness control (see \S \ref{sec:control_theta}) where the forcing amplitude $f(t)$ is given by \eqref{eq:momentum_control}. The distinction between the two considered numerical methods (i.e., DNS and ILES) introduces an extra subtlety in terms of how $f(t)$ must be chosen to achieve a desired unity value of the momentum thickness. For DNS, we use the explicit formula \eqref{eq:momentum_control}, letting
\begin{equation} \label{eq:f_DNS}
f(t) = f_{\text{DNS}}(t) := \frac{-ke_\theta(t)   - \frac12 C_f(t) + \frac{2}{\Rey} \int_0^\infty \left( \frac{ \partial U}{\partial y}\right)^2 dy - 2 \int_0^\infty \langle uv \rangle \frac{\partial U}{\partial y} \, dy}{\theta(t)}.
\end{equation}
Lemma \ref{thm:momentum_control} then guarantees that $\lim_{t \rightarrow \infty} \theta(t) =1$ for any solution to PDE \eqref{eq:NSE}, meaning that one would expect a DNS of \eqref{eq:ILES_1} to have the same behaviour. We show in \S \ref{sec:numerical_results} that this is indeed the case.

In the case of ILES it is not appropriate to  define the forcing amplitude by \eqref{eq:momentum_control} since this expression is only accurate for solutions of the Navier-Stokes equations \eqref{eq:NSE}, while ILES only gives an approximate solution via \eqref{eq:ILES_1}. This subtlety can be avoided by adding an `integral control' term to the forcing and using 
\[
f_\text{ILES}(t) := f_{\text{DNS}}(t) - k^2 \frac{z(t)}{\theta(t)}
\]
where $z(t) \in \mathbb{R}$ is defined as the integral of the momentum thickness via
\[
\frac{dz}{dt} = e_\theta(t) = 1- \theta(t). 
\]
This accounts, with the addition of only one extra scalar-valued variable, for the small difference between the theoretical value \eqref{eq:momentum_control} and the forcing required to maintain a unity value of momentum thickness in ILES. For simplicity, a gain value of $k=1$ is used throughout this paper. 

\subsection{Computational domains, data collection, and computational cost}

%

\begin{table} 
   \begin{center}
 \def~{\hphantom{0}}
   \begin{tabular}{lccccccc}
   $Re_\theta$ & Method & Domain Size  & Grid & $\Delta T^{+}$ & $T^{+}_{f}$ & Total Cost \\ \hline
~&~&($L_{x} \times L_{y} \times L_{z}$)/$\theta$ &($n_{x} \times n_{y} \times n_{z}$)&~&~&~(Core-hour) \\\hline
1000~&~DNS~&~$40\times30\times15$&~$192\times257\times128$~&~0.0214~&~4285~&~1280\\\hline
2000~&~DNS~&~$40\times30\times15$&~$352\times465\times224$~&~0.0177~&~7079~&~10752\\\hline
4060~&~ILES~&~$60\times30\times15$&~$320\times305\times128$~&~0.0482~&~12958~&~3584\\\hline 
6500~&~ILES~&~$80\times30\times15$&~$640\times449\times224$~&~0.0407~&~20347~&~31360\\\hline 
8300~&~ILES~&~$120\times30\times15$&~$1216\times545\times272$~&~0.0393~&~27538~&~163840 \\\hline 
\end{tabular}
   \caption{Simulation details for the DNS and ILES   of \eqref{eq:ILES_1}.}
   \label{tab:II}
   \end{center}
 \end{table}


The DNS are performed at $Re_{\theta}=1000, 2000$ and the ILES at $Re_{\theta}=4060, 6500, 8300$. These Reynolds numbers have been selected in order to allow for a direct comparison with published numerical data of a spatially evolving turbulent boundary layer. The spatial discretisation, temporal discretisation, and computational cost of each simulation run are given in Table \ref{tab:II}. With increasing Reynolds number it is necessary to use both a longer temporal window and a larger domain in the streamwise direction to obtain converged statistics. For example, $L_{x}$ is chosen to be $40/\theta$ for $\Rey_\theta=1000$ and is increased by a factor of three for $\Rey_\theta = 8300$. 

The flow statistics are computed by first performing a space/time average to obtain the mean streamwise velocity profile $U(y)$ and, subsequently, the wall shear stress $\tau_{w}=\nu U'(0)$ and friction velocity $U_{\tau}:=\sqrt{\tau_{w} / \rho} = \sqrt{\tau_w}$. Using $l_\ast=\nu / U_{\tau}$ as a length-scale and $t_\ast = l_\ast/ U_{\tau}$ as a time-scale, all variables can be expressed in wall units, e.g.,  $u^{+}=U/U_{\tau}$,  $y^{+}=y/l_\ast$, and $t^{+}=t/t_\ast$. In addition to the displacement thickness $\delta^\ast$ and momentum thickness $\theta$, we will also consider the `shape factor' $H_{12}:=\delta^\ast/\theta$, and the wake parameter
\begin{equation} \label{eq:cwake}
C_\text{wake} = \int_0^\infty (U^{+}_{\infty}- U^{+}) d(y/\delta) 
\end{equation}
introduced by \cite{Coles1956}, where  $\delta$ is the boundary layer thickness defined in terms of $99\%$ of the freestream velocity. 

The mean values of all flow statistics reported in subsequent sections are collected after the flow was observed to have transitioned to a statistically steady state. For the five Reynolds numbers considered, this was deemed to occur at non-dimensional times $t_0^+ = 1976, 3418, 7534, 12208$ and $16511$, ordered in terms of increasing $\Rey_\theta$. Each flow statistic was then averaged over a window $t^+ \in [t_0^+,T_{f}^+]$, respectively (see Table \ref{tab:II}), in order of increasing $\Rey_\theta$.
To describe statistical convergence, suppose that $g(t)$ is a given time-dependent flow property of interest, and $\bar{g}$ is its average over the window $[t_0^+,T_{f}^+]$. Letting
\[G(t^+) = \frac{1}{t^+-t_0^+} \int_{t_0^+}^{t^+} g(s) ds\] 
be the moving average of $g$, we define the maximum percentage deviation of the moving average from the reported mean over the final half the averaging window by
\[
E = \max \left\{ 100\% \cdot \left|\frac{G(t^+) - \bar{g}}{\bar{g}} \right| : \frac{ t_0^+ + T_f^+}{2} \leq t^+ \leq T_f^+ \right\}.
\]
The maximum value of $E$ for any of the flow statistics $\delta^\ast, U_\tau, f, C_\text{wake}$ that we report subsequently in Table \ref{tab:IV} is, $1.2\%,0.9\%,0.9\%,1.1\%$ and $1.5\%$, respectively, for the five cases considered in order of increasing $\Rey_\theta$.

Finally, we note that there are significant differences in computational cost, detailed in \ref{tab:II}, of the proposed numerical scheme in comparison to that of the highlighted reference studies listed in Table \ref{tab:III}. The DNS of \cite{Sillero2013}, achieving a maximum Reynolds number of $\Rey=6500$, is computationally expensive, and required 45 million core hours while using a large number of cores (32,768).  The LES of \cite{Eitel-Amor2014}, achieving a maximum Reynolds number of $\Rey_\theta = 8300$, used only around 1 million core hours and 4,096 cores. In this study, we sought a balance between computational cost and efficiency. At the highest Reynolds number considered in this study, $\Rey_\theta =8300$, the proposed numerical scheme used 160,384 core-hours on 8,192 cores, a six-fold reduction compared to the LES of \cite{Eitel-Amor2014}.


\section{Results} \label{sec:numerical_results}

In \S \ref{sec:verification} the DNS implementation of the numerical method proposed in this paper is validated against the DNS of \cite{Biau2023} at $\Rey_\theta=1000, 2000$. Subsequently, in \S \ref{sec:results_spatial_comparison}, an ILES implementation of the periodic boundary layer equations \eqref{eq:ILES_1} at $\Rey_\theta = 4060, 6500, 8300$ is compared with reference data from the DNS of \cite{Sillero2013} and \cite{Schlatter2010}, and with the LES of \cite{Eitel-Amor2014}. The reference data were obtained for spatially evolving turbulent boundary layers. 

A common colour scheme is used throughout \S \ref{sec:numerical_results} to indicate data from simulations at different Reynolds numbers:  ({\rule[0.5ex]{0.5cm}{1pt}}) black for $\Rey_\theta=1000$; ({\color{purple} {\rule[0.5ex]{0.5cm}{1pt}}}) purple for $\Rey_\theta=2000$; ({\color{blue} {\rule[0.5ex]{0.5cm}{1pt}}}) blue for $\Rey_\theta=4060$; ({\color{red} {\rule[0.5ex]{0.5cm}{1pt}}}) red for $\Rey_\theta=6500$; and ({\color{green!70!black} {\rule[0.5ex]{0.5cm}{1pt}}}) green for $\Rey_\theta=8300$. 

Results using the method presented in this paper are shown with solid lines ({\rule[0.5ex]{0.5cm}{1pt}}), and those from an implementation of the method of \cite{Biau2023} using \texttt{Xcompact3d} are shown with dotted lines (\hdashrule[0.5ex]{0.5cm}{1pt}{1pt}). Data extracted directly from \cite{Biau2023} are shown with ($\times$); those from the DNS of \cite{Schlatter2010} at $\Rey_\theta = 1000, 2000, 4060$ are shown with circles ($\ocircle$), while data from the DNS of \cite{Sillero2013} at $6500$ are shown with closed triangles ($\blacktriangledown$). The LES of \cite{Eitel-Amor2014} at $\Rey_\theta = 6500, 8300$ are also shown with circles ($\ocircle$), since on all subsequent figures these data are always visually distinguishable from those of \cite{Schlatter2010}. 

Table \ref{tab:III} also reports the grid resolutions in wall units between the present study and the reference works. For DNS, achieving accurate statistical results requires capturing a wide range of scales, e.g., from the boundary layer thickness to the Kolmogorov length scale.  Proper grid resolution is essential for this purpose. While it is ideal to fully resolve both scales, previous DNS studies have demonstrated that good agreement with experiments can still be achieved even when the Kolmogorov length scale remains partially under-resolved. This trade-off between accuracy and computational cost is noteworthy. \cite{Moin98} proposed  one possible set of grid resolutions, expressed in wall units as ($\Delta x^{+} = 14.3, \Delta y^{+} = 0.33, \Delta z^{+} = 4.8$), which seeks to strike this balance. Of the considered reference studies, the DNS of \cite{Sillero2013} used a grid resolution within the same range for both normal and spanwise directions, while \cite{Schlatter2010} used a fine spatial resolution in the normal direction, and a coarser grid in the streamwise and spanwise directions. The DNS performed in this paper have similar mesh resolutions to those in the reference data. For the ILES, following the work of \cite{Mahfoze2021}, the spatial resolution is more or less three times larger per spatial direction than in the DNS, while keeping a fine mesh close to the wall in the normal direction.


Finally, we note that the reference data are obtained for a spatially evolving turbulent boundary layer, and the reported wall units are typically scaled based on the $\Rey_{\theta}$ close to the end of the computational domain. In contrast, for the periodic turbulent simulations conducted here, $\Rey_{\theta}$ is constant throughout the simulation, which is advantageous since it allows a targeted choice of mesh resolution to be made {\it a priori}.

\begin{table} 
   \begin{center}
 \def~{\hphantom{0}}
   \begin{tabular}{lcccccccc}
   Case & Method & $\Delta x^{+}$ & $\Delta y^{+}_{min}$  & $\Delta z^{+}$ & $Re_{\theta}$ & Marker & Line Patterns\\ \hline
Present~&~DNS~&~10~&~0.5~&~5.5~&~-&~$\square$ &~Solid(-)\\\hline
Present~&~ILES~&~30~&~1.5~&~16.5~&~-&~$\square$ &~Solid(-)\\\hline
\cite{Schlatter2010}~&~DNS~&~18~&~0.35~&~9.6&~4300&~$\bigcirc$&~dashed(- -)\\\hline
\cite{Eitel-Amor2014}~&~LES~&~18~&~0.06~&~8~&~8300&~$\bigcirc$ &~dashed(- -)\\\hline
\cite{Sillero2013}~&~DNS~&~7~&~0.32~&~4.07~&~6500&~$\triangledown$&~dashdotted ($-\cdot$)\\\hline
\end{tabular}
   \caption{DNS and LES grid resolutions, and figure formatting conventions.}
   \label{tab:III}
   \end{center}
 \end{table}

\subsection{Validation against \cite{Biau2023}} \label{sec:verification}

We perform DNS by solving \eqref{eq:ILES_1} with $\Gamma=0$ and with the forcing amplitude $f$, given by \eqref{eq:f_DNS}, at $\Rey_\theta=1000, 2000$. Data are compared with those extracted from \cite{Biau2023} (obtained with a different flow solver), and with a separate implementation of Biau's numerical scheme (in which $\theta=1$ is imposed at each time step via the iterative solution of an optimisation problem).

\begin{figure}
\centering
\includegraphics[width=1.0\linewidth, trim={0.cm 0cm 0.cm 0.0cm},clip]{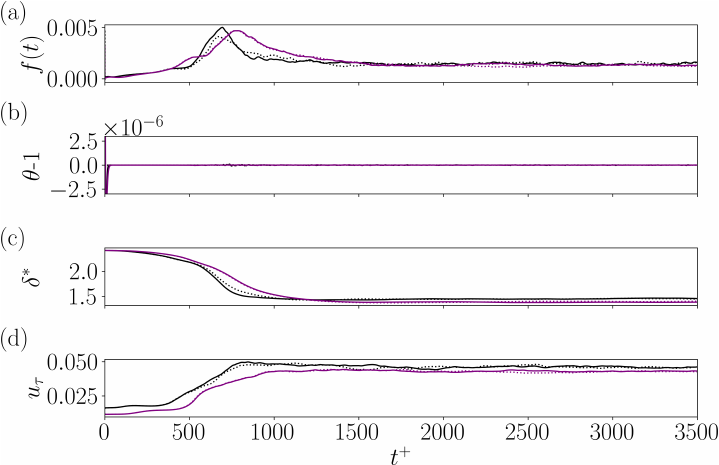}
\caption{Temporal variation of the forcing amplitude $f$, momentum thickness $\theta$, displacement thickness $\delta^\ast$ and friction velocity $u_\tau$. Data from the present DNS using $f_\text{DNS}$ are shown with solid lines ({\rule[0.5ex]{0.5cm}{1pt}}$\Rey_\theta=1000$;  {\color{purple}\rule[0.5ex]{0.5cm}{1pt} }$\Rey_\theta=2000$)
 and from implementation of Biau's method in \texttt{Xcompact3d} with dotted lines (\hdashrule[0.5ex]{0.5cm}{1pt}{1pt}$\Rey_\theta=1000$;  {\color{purple}\hdashrule[0.5ex]{0.5cm}{1pt}{1pt} }$\Rey_\theta=2000$). }
\label{fig:V}
\end{figure}

The temporal evolution of $f,\theta, \delta^\ast$ and $u_\tau$ is shown in figure  \ref{fig:V}. The instantaneous friction velocity $u_\tau$ is defined in terms of the wall normal derivative of the streamwise- and spanwise-averaged wall normal velocity,  
\[
u_\tau = \nu \frac{\partial U}{\partial y}(0,t).
\]
For the two implementations of Biau's method, the momentum thickness satisfies $\theta =1$ at all times. The present method has the same behaviour as the method initially designed by Biau, apart from a small initial transient corresponding to an exponential convergence of the error $e_\theta \rightarrow 0$. For all three methods (Biau original method with his solver, Biau original method with \texttt{Xcompact3d}, and the present method with \texttt{Xcompact3d}), $f$ approaches a statistically steady state with mean values of approximately $0.0012$ for $\Rey_\theta = 1000$, and $0.0014$ for $\Rey_\theta=2000$. This confirms the veracity the analytical expression for the asymptotic value of the forcing term \eqref{eq:f_asymptotic_momentum} in this case. 

It is interesting to observe that, for all methods, the displacement thickness $\delta^\ast$, whose value is not prescribed, also converges to a constant. This constant is observed to decrease with increasing Reynolds number, implying a corresponding decrease in shape factor $H_{12} = \delta^\ast / \theta$, which is in line with the observations presented in \cite{Schlatter2010}. Similarly, the friction velocity $u_\tau$ also decreases with increasing Reynolds number.  Overall, the relative error between any of the three methods for any of the flow statistics $U_{\tau}, \delta, \delta^\ast, \theta, H_{12}, C_{f}$, and $C_\text{wake}$ does not deviate by more than $1\%$, indicating excellent agreement between the current implementation and with the data presented in  \cite{Biau2023}.

\begin{figure}
\centering
\includegraphics[width=1\linewidth, trim={0.cm 0cm 0.cm 0.0cm},clip]{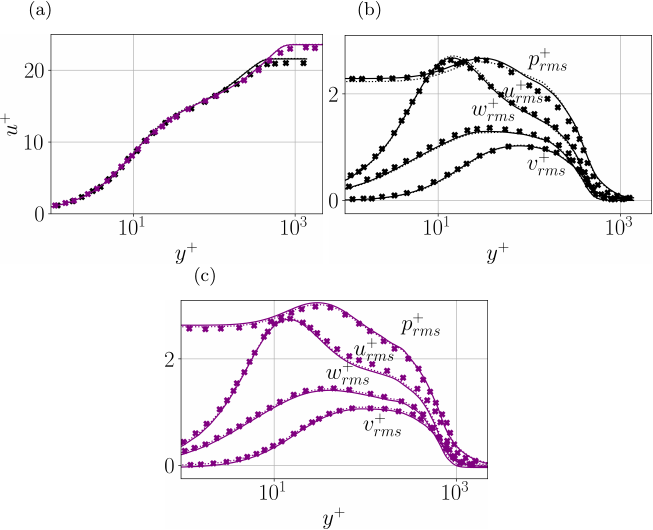}
\captionsetup{justification=centering}
\caption{A comparison of DNS of the current method with that of \cite{Biau2023}: (a) Mean streamwise velocity $u^+$; (b) rms velocities and pressures at $Re_\theta=1000$; (c) rms velocities and pressures at $Re_\theta=2000$.  Data from \cite{Biau2023} are shown with markers ($\mathbf{\times}$~$\Rey_\theta=1000$; {\color{purple} $\mathbf{\times}$}~$\Rey_\theta=2000$), from the present DNS using $f_\text{DNS}$ with solid lines  ({\rule[0.5ex]{0.5cm}{1pt}}$\Rey_\theta=1000$;  {\color{purple}\rule[0.5ex]{0.5cm}{1pt} }$\Rey_\theta=2000$), and from the implementation of Biau's method in \texttt{Xcompact3d} with dotted lines (\hdashrule[0.5ex]{0.5cm}{1pt}{1pt}$\Rey_\theta=1000$;  {\color{purple}\hdashrule[0.5ex]{0.5cm}{1pt}{1pt} }$\Rey_\theta=2000$). }
\label{fig:IV}
\end{figure}

Figure \ref{fig:IV} shows the mean stream wise velocity $u^+$ and root-mean-square (rms) velocity and pressure profiles obtained from the three methods. The  implementation of Biau's method in \texttt{Xcompact3d} compares very well with the present method. Both exhibit small, but noticeable, differences from data extracted directly from \cite{Biau2023} which are most prominent in  $u^{+}$ and $p_{rms}^{+}$ in the outer layer for $\Rey_\theta=1000$. These discrepancies are lower at $\Rey_\theta=2000$ and it is important to note that at this Reynolds number the implementation of the present method in \texttt{Xcompact3d} is shown in \S \ref{sec:spatial_comp_velocities} to compare very well with reference DNS data from spatially evolving boundary layers.

 \subsection{Comparison against spatially evolving turbulent boundary layer data} \label{sec:results_spatial_comparison}

\subsubsection{Statistical quantities}

\begin{table}
   \begin{center}
   \def~{\hphantom{0}}
   \resizebox{\textwidth}{!}{ 
   \begin{tabular}{l|ccccccccc}
   $ Case $     &~$Re_\theta$    &~$Re_{\tau}$&~$Re_{\delta^{*}}$&~$U_{\tau}$&~$H_{12}$&~$C_{f}$&~$f(t)$&~$C_{wake}$ &~$\delta^{*}$ \\[2pt]\hline
Present&~1000&~398&~1448&~0.0463&~1.448&~0.00429&~0.00141&~3.578&~1.448 \\\hline
\cite{Schlatter2010}&~1006&~359&~1459&~0.0462&~1.445&~0.00426&~-&~4.062&~1.451\\\hline
\cite{Sillero2013}&~1100&~445&~1585&~0.0462&~1.434&~0.00426&~-&~3.586&~1.434\\\hline\hline
Present&~2000&~712&~2806&~0.0421&~1.403&~0.00359&~0.00121&~3.938&~1.403\\\hline
\cite{Schlatter2010}&~2000&~671&~2827&~0.0421&~1.414&~0.00353&~-&~4.214&~1.414\\\hline
\cite{Sillero2013}&~1968&~690&~2780&~0.0422&~1.416&~0.00356&~-&~3.984&~1.418\\\hline\hline
Present&~4060&~1315&~5614&~0.0385&~1.383&~0.00298&~0.00102&~4.269&~1.383\\\hline
\cite{Schlatter2010}&~4061&~1271&~5633&~0.0385&~1.387&~0.00297&~-&~4.431&~1.387\\\hline
\cite{Sillero2013}&~4000&~1306&~5589&~0.0390&~1.377&~0.00304&~-&~4.324&~1.375\\\hline\hline
Present&~6500&~2141&~8749&~0.0373&~1.346&~0.00279&~0.00097&~4.085&~1.346\\\hline
\cite{Eitel-Amor2014}&~6500&~1972&~8886&~0.0368&~1.367&~0.00273&~-&~4.494&~1.364\\\hline
\cite{Sillero2013}&~6500&~1989&~8879&~0.0368&~1.363&~0.00270&~-&~4.492&~1.363\\\hline\hline
Present&~8300&~2763&~10987&~0.0368&~1.324&~0.00272&~0.00095&~3.976&~1.324\\\hline
\cite{Eitel-Amor2014}&~8300&~2557&~11192&~0.0364&~1.348&~0.00265&~-&~4.464&~1.352\\\hline
\end{tabular}}
\caption{A comparison of flow statistics between periodic boundary layer and spatially evolving boundary layer simulations.}
\label{tab:IV}
\end{center}
\end{table}

Table \ref{tab:IV} presents some key statistics for the turbulent boundary layer flows simulated with the present method and those of the selected reference data for spatially evolving turbulent boundary layers. It should be emphasised that for the present method the fixed control parameter is $\Rey_\theta$, while the other reported statistics are emergent properties of the flow. Furthermore, throughout this section the friction Reynolds number is defined in terms of the $99\%$ boundary layer thickness via
$ \Rey_\tau = U_\tau \delta / \nu$.

The wake statistics for the present method compare very well with the reference data, with $H_{12}, \Rey_{\delta^{\ast}}$ and $U_\tau$  not deviating by more than $2\%$ from any of the corresponding reference data values, for any of the considered Reynolds numbers $\Rey_\theta$. The skin friction coefficient $C_f$ also compares very well with the reference data, with a maximum deviation of only $3.3\%$ from the results of \cite{Sillero2013} at $\Rey_\theta=6500$. 

Slightly larger discrepancies are observed for the friction Reynolds number $\Rey_\tau$ and the wake coefficient $C_\text{wake}$. These  have  maximum deviations of $10.8\%$ and $11.9\%$, respectively, from the DNS of \cite{Schlatter2010} at $\Rey_\theta=1000$, although these deviations decrease with Reynolds number and are below $4\%$ for $\Rey_\theta = 4060$. The larger discrepancies for $\Rey_\tau$ and $C_\text{wake}$ can be explained by the fact that both are defined in terms of the the $99\%$ boundary layer thickness $\delta$, as opposed to $H_{12}, \Rey_{\delta^\ast}$ which are defined using the displacement thickness $\delta^\ast$. Since $\delta$ is a pointwise measure of the boundary layer thickness, this is a less robust statistic than the integral quantity $\delta^\ast$. This is important, given that the forcing $\bs{f} \sim y \frac{\partial \bs{u}}{\partial y}$ required to maintain the momentum thickness is of largest magnitude away from the wall. 

The above observations are confirmed by figure \ref{fig:XIII} which shows excellent agreement between the present method and the reference data in terms of the scaling of displacement Reynolds number (in figure \ref{fig:XIII}~(a)), and the  skin friction coefficient and shape factor (in figure \ref{fig:XIII}~(b)) with $\Rey_\theta$. Slightly larger deviations for the reference data can be observed in figure \ref{fig:XIII}~(a) in terms of the scaling of $\Rey_\tau$ with $\Rey_\theta$. Despite this, numerical fits (shown in the legend of figure \ref{fig:XIII}~(a)) to the data reveal near-linear scalings of $\Rey_\tau$ and $\Rey_{\delta^\ast}$ with the control parameter $\Rey_\theta$, namely 
\[
\Rey_\tau \approx \alpha_1 \Rey_\theta^{\beta_1} \quad \text{and} \quad \Rey_{\delta^\ast} \approx \alpha_2 \Rey_\theta^{\beta_2},
\]
for prefactors in the ranges $0.41 \leq \alpha_1 \leq 1$ and $ 1.82 \leq \alpha_2 \leq 2.84$ and scaling exponents in the ranges $0.87 \leq \beta_1 \leq 0.97$ and $ 0.95 \leq \beta_2 \leq 0.97$. The smaller range of scaling exponents for $\beta_2$ again confirms that $\Rey_{\delta^\ast}$ is a more robust statistic than $\Rey_\tau$. 

These numerical fits are consistent with the theoretical scaling laws proven in \S \ref{sec:bounds}, which state that the scaling exponents cannot exceed $\beta =1$. It is still an open theoretical question whether the small observed deviations of the scaling exponents $\beta_1,\beta_2$ from this value arise due to a logarithmic term in the true scaling law for the skin friction in  turbulent boundary layers. 

\begin{figure}
\centering
\includegraphics[width=1\linewidth, trim={0.cm 0cm 0.cm 0.0cm},clip]{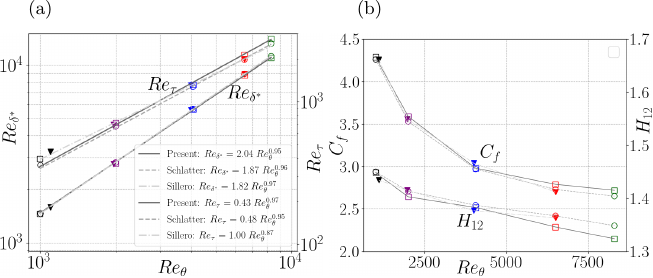}
\captionsetup{justification=centering}
\caption{The scaling of (a) $\Rey_{\tau}$ and $\Rey_{\delta^\ast}$; and (b) $C_f$ and $H_{12}$ with   $\Rey_{\theta}$. The markers indicate: ($\ocircle$) for both  DNS data from \cite{Schlatter2010}'s and LES data from \cite{Eitel-Amor2014}'s;  ($\blacktriangledown$) for DNS data from \cite{Sillero2013}; and ($\square$) data from the present method. }
\label{fig:XIII}
\end{figure}

\subsubsection{Velocity profiles, fluctuating statistics, and TKE budgets} \label{sec:spatial_comp_velocities}

\begin{figure}
\centering
\includegraphics[width=1\linewidth, trim={0.cm 0cm 0.cm 0.0cm},clip]{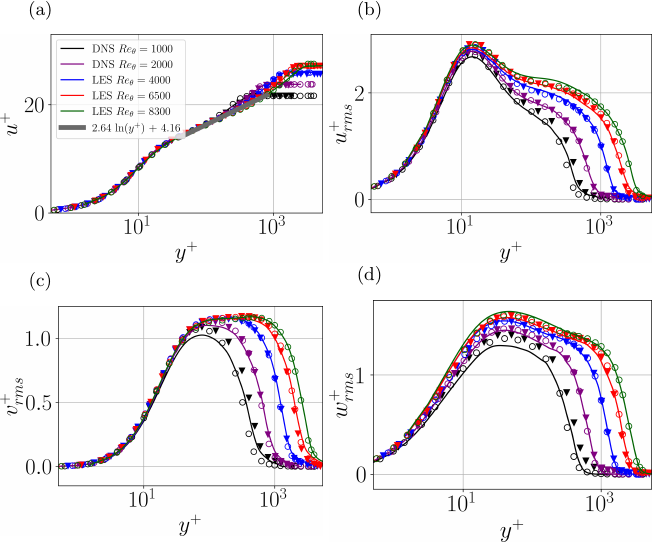}
\caption{Profiles of (a) $u^+$; (b) $u_{rms}^+$; (c) $v_{rms}^+$; and (d) $w_{rms}^{+}$. Reference data are shown with markers: ($\ocircle$) for both the DNS of \cite{Schlatter2010} and the LES of  \cite{Eitel-Amor2014}; ($\blacktriangledown$) for the DNS of \cite{Sillero2013}. Results with the present method are shown with solid lines. The colour convention is explained in \S \ref{sec:numerical_results}.}
\label{fig:VI}
\end{figure}

Figure \ref{fig:VI}~(a) shows the mean velocity profile $u^+$ in comparison with the reference data from spatially evolving boundary layer simulations at comparable Reynolds numbers. Excellent agreement can be observed in the inner layer, with only small deviations in the wake region. In the linear and logarithmic regions, a grey line shows the fit $\kappa^{-1} \ln{y^{+}} + B$ with $\kappa =0.379$ and $B=4.161$ to the present data, which is in good agreement with that presented in \cite{Eitel-Amor2014}. As $Re_{\theta}$ increases, the asymptotic value of $u^+$ in the wake region increases from $u^{+}=21.60$ at $Re_{\theta}=1000$ to $u^{+}=27.13$ at $Re_{\theta}=8300$, indicated by the arrow in Figure \ref{fig:VI} (a), with very good agreement between the reference data and the present method.

Figures \ref{fig:VI}~(b-d) show profiles of rms velocity fluctuations. Again, these reveal a good agreement between data from periodic and spatially evolving simulations. The peak values of the streamwise velocity fluctuations $u^{+}_{rms}$ are observed to lie consistently in the buffer layer at $y^{+} \approx 14$, which agrees well with the value of  $y^{+} \approx 15$ reported by both \cite{Devenport2022} and \cite{Smits2021}. Following this peak value $u^{+}_{rms}$ decreases, with the rate of decrease slowing in the overlap region near $y^{+} \approx 100$. Here, the gradient $\partial u^{+}_{rms} / \partial y^{+}$ is observed to increase with $\Rey_\theta$, with data from the highest considered Reynolds number $\Rey_\theta = 8300$ exhibiting a plateau consistent with the experimental observations of \cite{Devenport2022} for spatially evolving boundary layers. 

Regarding peak values of the rms velocities, \cite{Smits2021} reported that the peak value of the squared rms streamwise velocity, denoted $u_{rms}^{2+}$, was observed to lie on the line $3.54 + 0.646 \ln(Re_{\tau})$ for friction Reynolds numbers in the range $6123 \leq Re_{\tau} \leq 19680$. Although our studies are at lower Reynolds numbers, we observe a similar relationship of $3.88 + 0.56 \ln(Re_{\tau})$. This compares well to analogous fits to the combined data of \cite{Schlatter2010,Eitel-Amor2014} and to the data of \cite{Sillero2013}, revealing relationships of $3.96 + 0.58 \ln(Re_{\tau})$ and $3.15 + 0.71 \ln(Re_{\tau})$, respectively. We note that for this analysis, the data of \cite{Schlatter2010} and \cite{Eitel-Amor2014}, which were produced by the same group and the same code, were combined to achieve a wider range of Reynolds numbers. Table \ref{tab:alpha_beta} also shows that the location of the peak values of $v_{rms}^{2+}$ and $w_{rms}^{2+}$ also lie on the lines of the form $\alpha \ln{(\Rey_\tau)} + \beta$ with a good degree of confidence, and that there is good agreement between the present results and data from the considered reference studies.



\begin{table}
\centering
\vskip 0.5em
\begin{tabular}{p{4cm}ccccccccc}
   & \multicolumn{3}{c}{$u^{2+}_{rms}$} & \multicolumn{3}{c}{$v^{2+}_{rms}$} &  \multicolumn{3}{c}{$w^{2+}_{rms}$} \\
\cmidrule(lr){2-4}\cmidrule(lr){5-7}\cmidrule(lr){8-10}
 & $\alpha$ & $\beta$   & $R^2$  &   $\alpha$ & $\beta$   & $R^2$  &   $\alpha$ & $\beta$   & $R^2$  \\
  \hline
  Present & $0.56$ & $3.88$ & $0.973$& $0.1$ & $0.57$ & $0.875$& $0.4$ & $-0.51$ & $0.982$ \\
  \hline
\cite{Schlatter2010}  and \cite{Eitel-Amor2014} & $0.58$ & $3.96$ & $0.996$& $0.1$ & $0.6$ & $0.918$& $0.37$ & $-0.26$ & $0.996$ \\
  \hline
  \cite{Sillero2013} & $0.71$ & $3.15$ & $0.994$& $0.12$ & $0.48$ & $0.973$ & $0.38$ & $-0.3$ & $0.999$\\
  \hline
\end{tabular}
\caption{ Linear regression coefficients $\alpha,\beta$ and correlation statistics $R^2$ for fits of the peak value of the squared rms velocity, $u^{2+}_{rms},v^{2+}_{rms}$ and $w^{2+}_{rms}$ to the line $\alpha \ln{(\Rey_\tau)} + \beta$. Fits are reported to data from the present study; the combined data of \cite{Schlatter2010,Eitel-Amor2014}; and the data of \cite{Sillero2013}.}
\label{tab:alpha_beta}
\end{table}


Regarding differences between periodic boundary layers and the spatially evolving reference data, the most prominent are in the case $\Rey_\theta =1000$ where differences are apparent in all three rms velocity profiles, as well as in the Reynolds shear stress and rms pressure profiles shown in figure \ref{fig:VII}. However, it can be seen in figures \ref{fig:VI} and \ref{fig:VII} that these differences decrease significantly with $\Rey_\theta$, a finding that is consistent with the analysis of \cite{Kozul2016} for temporal, periodic, boundary layers. 
A slight deviation is observed between the  $u^{+}_{rms}$ profile obtained in the present study and that of \cite{Eitel-Amor2014} for $\Rey_\theta=8300$. A possible cause of this discrepancy is that in this simulation of a spatially-evolving boundary layer, the spatial location for the statistics for $\Rey_\theta =8300$ is very close to the outlet of the computational domain, with \cite{Eitel-Amor2014} only technically reporting results on $u^{+}$ up to $\Rey_\theta=7500$. At such a location, the outlet boundary conditions can cause spurious, nonphysical behaviour. Thus, while we believe it is useful to consider the data from \cite{Eitel-Amor2014} at $\Rey_\theta =8300$ for comparison, they should be treated with appropriate caution. It was reported in  \cite{Fernholz1996} that the peak location of the Reynolds shear stress is proportional to $\sqrt{Re_\tau}$. Our observations support this, revealing a strong correlation ($R^2=0.984$) of this peak value with the line $2.24 \sqrt{Re_{\tau}}+0.97$. A similar analysis using the combined data from \cite{Schlatter2010,Eitel-Amor2014} and \cite{Sillero2013}, yields fits of $2.48\sqrt{Re_{\tau}}-0.1$ ($R^2 = 0.847$) and of $2.99 \sqrt{Re_{\tau}}+19.35$ ($R^2 = 0.951$), respectively, and show proportionality constants in good agreement with the data using the present method.  Finally, we note that Figure \ref{fig:VII} shows discrepancies in $p_{rms}^+$ near the wall between the present method and the reference data, which can be attributed to the coarser near-wall grid resolution employed in this study (see Table \ref{tab:III}).

\begin{figure}
\centering
\includegraphics[width=1\linewidth, trim={0.cm 0cm 0.cm 0.0cm},clip]{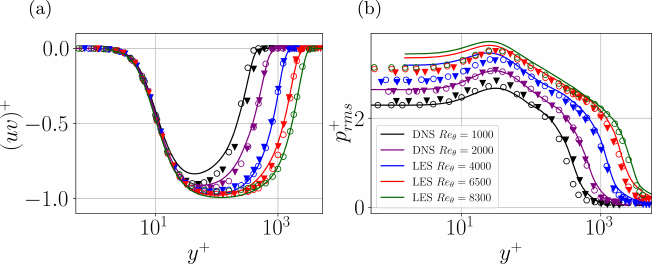}
\caption{(a) Mean Reynolds shear stress; and (b) rms pressure profiles. Reference data are shown with markers: ($\ocircle$) for both the DNS of \cite{Schlatter2010} and the LES of  \cite{Eitel-Amor2014}; ($\blacktriangledown$) for the DNS of \cite{Sillero2013}. Results with the present method are shown with solid lines.}
\label{fig:VII}
\end{figure}



Figure \ref{fig:TKE} shows the mean turbulent kinetic energy budget for $\Rey_\theta=1000,~6500$. It can be seen that the largest difference between periodic and spatially evolving boundary layers are obtained for the lowest Reynolds number. In particular, the production, dissipation and viscous diffusion components of the budget have consistently lower magnitudes in the periodic case when compared to the spatially evolving case. For the highest Reynolds number in the figure the present energy budget is in excellent agreement with the reference data.

\begin{figure}
\centering
\includegraphics[width=1\linewidth, trim={0.cm 0cm 0.cm 0.0cm},clip]{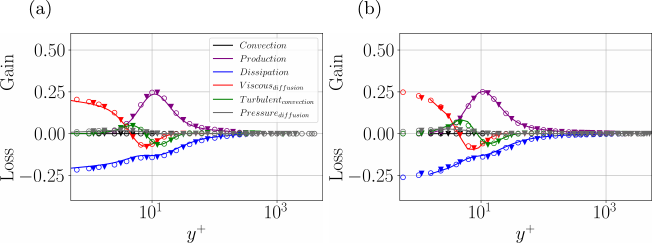}
\caption{Turbulent Kinetic Energy (TKE) budgets for (a) $\Rey_\theta=1000$ and (b) $\Rey_\theta = 6500$. Reference data are shown with markers: ($\ocircle$) for both the DNS of \cite{Schlatter2010} and the LES of  \cite{Eitel-Amor2014}'s; ($\blacktriangledown$) for the DNS of \cite{Sillero2013}. Results with the present method are shown with solid lines. }
\label{fig:TKE}
\end{figure}

\subsubsection{Vorticity profiles and instantaneous visualisations}

Figure \ref{fig:XI} shows the profiles of the rms vorticity components. Overall, there is very good agreement between the present method and the reference data, in terms of the locations of minimum and maximum rms vorticity, as well as the general profile shapes. 
%
%
The peak location of $\omega_{y_{rms}}$ is in the buffer layer at  $y^{+} \approx 13$, which is consistent with the results of \cite{Schlatter2010,Eitel-Amor2014} and \cite{Sillero2013}.

\begin{figure}
\centering
\includegraphics[width=1\linewidth, trim={0.cm 0cm 0.cm 0.0cm},clip]{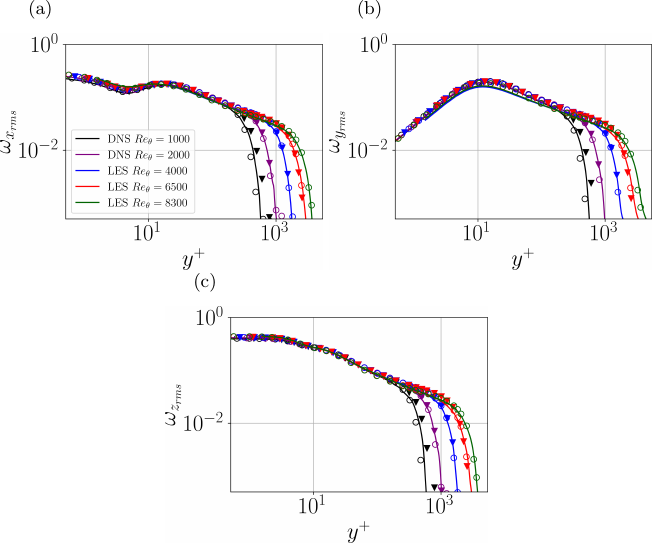}
  \caption{Profiles of the rms vorticity components (a) $\omega_{x_{rms}}$; (b) $\omega_{y_{rms}}$; and (c) $\omega_{z_{rms}}$, for $Re_\theta=1000$ to $8300$. Reference data are shown with markers: ($\ocircle$) for both the DNS of \cite{Schlatter2010} and the LES of  \cite{Eitel-Amor2014}'s; ($\blacktriangledown$) for the DNS of \cite{Sillero2013}. Results with the present method are shown with solid lines.}
  \label{fig:XI}
\end{figure}

Finally, we use the Q-criterion to visualise the vorticial structures generated in our periodic turbulent boundary layer flows. The $Q$-criterion can be used to measure the local balance of rotation and strain rate, being defined by $Q = \frac12 ( \| \Omega \|_F^2 - \|S\|_F^2 )$, where $S=\frac12 \left( \partial u_i / \partial x_j + \partial u_j / \partial x_i\right)$ is the stain rate tensor, $\Omega=\frac12 \left( \partial u_i / \partial x_j - \partial u_j / \partial x_i\right)$ is the vorticity tensor,  and $\| \cdot \|_F$ is the Frobinus norm. The Q-criterion is normalized using wall units as  $Q^{+} = Q(t^+)^{-2} = Q \nu^2 (U_{\tau})^{-4}$. The normalized Q-criterion plots, which show contours at the values of $Q^+ = \{0.01, 0.004\}$ at Reynolds numbers $\Rey_{\theta} = \{1000, 6500 \}$, respectively, are presented in Figure \ref{fig:Q_series}. 
As expected, a wider range of turbulent scales can be observed for the DNS than the ILES, with ﬁner vortices apparent as the Reynolds number is increased. It should be noted that the ILES does not produce spurious oscillations, suggesting that the numerical dissipation is acting properly. The shape and structure of the vortices (mainly elongated structures in the streamwise direction, slightly inclined up with respect to the wall) is similar to the ones observed in spatially evolving turbulent boundary layers. It is clear that the periodic boundary layer model considered in this paper appears to capture the expected features of canonical turbulent boundary layers across a range of Reynolds numbers.

\begin{figure}
\centering
\includegraphics[width=1\linewidth, trim={0.cm 0cm 0.cm 0.0cm},clip]{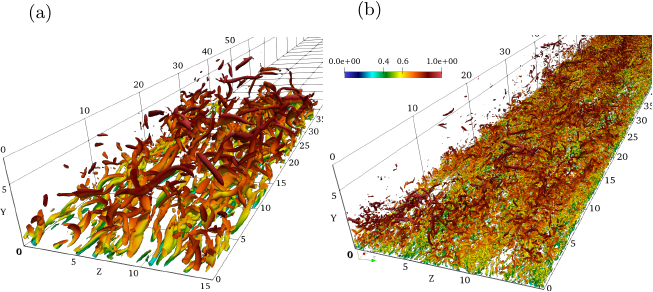}
\caption{Contours of constant $Q^+$ for selected snapshots of a periodic turbulent boundary layer at (a) $\Rey_\theta =1000, Q^+ = 0.01$;  (b)  $\Rey_\theta =6500, Q^+ = 0.004$. The colourbar indicates non-dimensional streamwise velocity $u$. Each subplot shows a section of the respective spatial domains described in Table \ref{tab:II}.}
\label{fig:Q_series}
\end{figure}

\section{Conclusions} \label{sec:conclusion}

In this paper we have performed a rigorous analytical study of the periodic boundary layer equations proposed by \cite{Biau2023}, in which a body force is used to maintain the boundary layer thickness of the flow. It is shown that an explicit formula can be obtained for the amplitude of this body force as a function of the flow velocity field. This enabled an explicit PDE to be identified for the flow, as opposed to the implicit definition given by \cite{Biau2023}. 

The explicit form of the PDE was important for two reasons. First, it allowed an application of the `background field' method of and Constantin-Doering, and we proved that the skin friction coefficient of the periodic boundary layer flow was upper bounded by an absolute constant. Future work will investigate wither this constant can be reduced by implementing computationally an optimal version of the proof presented in this paper. 

The second implication of our explicit formula for the forcing amplitude was that it allowed the construction of a simple numerical scheme for simulating turbulent boundary layers on periodic domains which can be used with both DNS and ILES approaches. Validation of this scheme was presented, with results shown to closely match those from \cite{Biau2023}, but also with data from spatially evolving turbulent boundary layers up to $\Rey_\theta = 8300$. It was observed that the similarity between the two classes of flow was greater at higher Reynolds number. This is important, since  periodic turbulent boundary layer simulations can be performed at a lower computational cost than simulations of spatially evolving boundary layers. For example, the ILES  of a periodic boundary layer flow at $\Rey_\theta=8300$ performed in this study is nearly 300 times cheaper than the DNS of \cite{Sillero2013}, and about 6 times cheaper than the LES of \cite{Eitel-Amor2014} for spatially evolving boundary layers at comparable Reynolds numbers.




In conclusion, this paper presented a detailed theoretical and numerical study of turbulent boundary layer flows and investigated their statistical similarity with canonical spatially evolving boundary layers.  Further investigations at higher Reynolds numbers and exploration of more complex flow scenarios are warranted to fully explore  our findings.

\section{Acknowledgements}
This work was funded via the EPSRC project EP/T021144/1. The simulations were performed on the ARCHER2 UK National Supercomputing Service via a Pioneer project and the UK Turbulence Consortium (EP/R029326/1). Saeed Parvar would like to thank Ms. Felicity Buchan, Member of Parliament of the United Kingdom, for her support and encouragement.

\section{Appendix} \label{sec:appendix}
In this section we provide detailed proofs of Lemma \ref{lem:displacement_control}, Lemma \ref{thm:momentum_control} and the energy equation \eqref{eq:displacement_control_E}.

\subsection*{Proof of \eqref{eq:displacement_evolution} and Lemma \ref{lem:displacement_control}} 

\begin{proof}
By taking the streamwise and spanwise average of \eqref{eq:NSE_d1}, it can be shown that 
\begin{equation} \label{eq:U_pde}
\frac{\partial U}{\partial t} + \frac{\partial}{\partial y} \langle u v\rangle = \frac{1}{\Rey} \frac{\partial^2 U}{\partial y^2} + f(t) y \frac{\partial U}{\partial y}.
\end{equation}
Using the boundary conditions, it follows that
\[
\frac{d \delta^\ast}{dt} = - \int_0^\infty \frac{\partial U}{\partial t} dy = \frac{1}{\Rey} \frac{\partial U}{\partial y}(0,t) - f(t) \int_0^\infty y  \frac{\partial U}{\partial y} dy.
\]
We next apply integration by parts to the final term and use the assumption that $y(1-u(y))\rightarrow 0$ as $y \rightarrow \infty$ to obtain
\begin{equation} \label{eq:dUdy_formula}
\int_0^\infty y  \frac{\partial U}{\partial y} dy = \int_0^\infty y  \frac{\partial (U-1)}{\partial y} dy = \left[ y(U-1)\right]_{y=0}^\infty  + \int_0^\infty (1- U) dy = \delta^\ast
\end{equation}
Hence, 
\begin{equation} \label{eq:delta_ode}
\frac{d \delta^\ast}{dt} = -f(t)\delta^\ast(t) + \frac{1}{\Rey} \frac{\partial U}{\partial y}(0,t).
\end{equation}
Letting $e_\delta(t) = 1- \delta^\ast(t)$ and using the control law 
\[
f(t) = \frac{-k e_\delta(t) + \frac{1}{\Rey}\frac{\partial U}{\partial y}}{\delta^\ast(t)},
\]
it  follows that $\dot{e}_\delta = -ke_\delta$. Hence, $e_\delta(t) \rightarrow 0$ as $t \rightarrow \infty$. 
\end{proof}

\subsection*{Proof of \eqref{eq:momentum_ode} and Lemma \ref{thm:momentum_control}.} 

\begin{proof}
Using \eqref{eq:U_pde}, 
\begin{align}
   \frac{d \theta}{dt} &= \int_0^\infty \frac{\partial U}{\partial t} - 2 U  \frac{\partial U}{\partial t} dy \nonumber \\
   &=- \frac{d\delta^\ast}{dt} + 2\int_0^\infty U \frac{\partial \langle uv \rangle}{\partial y} -\frac{2}{\Rey} \int_0^\infty U \frac{\partial^2 U}{\partial y^2} dy - 2 f(t) \int_0^\infty yU \frac{\partial U}{\partial y} dy \nonumber \\
   &= f(t) \delta^\ast(t) - \frac{1}{\Rey} \frac{\partial U}{\partial y}(0,t) \nonumber \\
   &\quad - 2 \int_0^\infty \frac{\partial U}{\partial y} \langle uv \rangle dy + \frac{2}{\Rey} \int_0^\infty \left( \frac{\partial U}{\partial y} \right)^2 dy - 2 f(t) \int_0^\infty yU \frac{\partial U}{\partial y} dy, \label{eq:dtheta_dt}
\end{align}
where in the final line, we have used integration by parts,  the boundary conditions \eqref{eq:bcs}, and the identity \eqref{eq:delta_ode}. To evaluate the final term, we again use integration by parts and the assumption that $y(1-U(y)) \rightarrow 0$ as $y \rightarrow \infty$ to obtain
\begin{align}
  \int_0^\infty yU \frac{\partial U}{\partial y} dy  &=   \int_0^\infty yU \frac{\partial( U-1)}{\partial y} dy \nonumber\\
  &=  \left[ yU(U-1) \right]_{y=0}^{y=\infty} - \int_0^\infty U(U-1) dy - \int_0^\infty y(U-1) \frac{\partial U}{\partial y} dy \nonumber \\
\text{(by \eqref{eq:dUdy_formula})}  &=\theta + \delta -  \int_0^\infty yU \frac{\partial U}{\partial y} dy. \label{eq:UdUdy_formula}
\end{align}
Combining \eqref{eq:dtheta_dt} and \eqref{eq:UdUdy_formula} gives
\[
\frac{d \theta}{dt}  = - f(t)\theta(t) - \frac{1}{\Rey} \frac{\partial U}{\partial y}(0,t)  - 2 \int_0^\infty \frac{\partial U}{\partial y} \langle uv \rangle dy + \frac{2}{\Rey} \int_0^\infty \left( \frac{\partial U}{\partial y} \right)^2 dy.
\]
Hence, if $e_\theta(t) = 1-\theta(t)$ and $f(t)$ is given by \eqref{eq:momentum_control}, it follows that 
\[
\frac{de_{\theta}}{dt} = -ke_\theta(t)
\]
and, hence, $e_\theta(t) \rightarrow 0$ as $t \rightarrow \infty$. 
\end{proof}

\subsection*{Proof of the energy identity \eqref{eq:displacement_control_E}}

\begin{proof}
    After taking the dot product of \eqref{eq:NSE_d1} with $\bs{u}$, integrating over $\Omega$, and using incompressibility gives
\begin{equation} \label{eq:dudt_1}
\frac12 \frac{d}{dt} \|\bs{u}\|^2 + \frac{1}{\Rey} \|\nabla \bs{u}\|^2  = \frac12 C_f(t) \int_\Omega y \bs{u} \cdot \frac{\partial \bs{u}}{\partial y} d\bs{x}. 
\end{equation}
Now consider the final term in the above equation. An initial integration by parts gives 
\begin{equation} \label{eq:dudt_2}
\int_\Omega y \bs{u} \cdot \frac{\partial \bs{u}}{\partial y} d\bs{x} = - \frac12 (\|v\|^2 + \|w\|^2) + \int_\Omega y u  \frac{\partial u}{\partial y} d\bs{x}. 
\end{equation}
Now, \eqref{eq:dUdy_formula} implies that $\delta^\ast = \int_0^\infty y \frac{\partial U}{\partial y} dy = \int_\Omega y\frac{\partial u}{\partial y} \, d\bs{x}$. It then follows that 
\begin{align}
 \int_\Omega y u  \frac{\partial u}{\partial y} d\bs{x} &= \delta^\ast + \int_\Omega y(u-1) \frac{\partial u}{\partial y} \, d\bs{x} \nonumber \\
 &= \delta^\ast + \int_\Omega y(u-1) \frac{\partial (u-1)}{\partial y} \, d\bs{x}\nonumber \\
 \text{(by \eqref{eq:decay})} &=\delta^\ast - \|u-1\|^2 - \int_\Omega y(u-1) \frac{\partial u}{\partial y} \, d\bs{x} \nonumber \\
 & = 2\delta^\ast - \|u-1\|^2 -  \int_\Omega y u  \frac{\partial u}{\partial y} d\bs{x}. \label{eq:yududy_formula}
\end{align}
Finally, combining \eqref{eq:dudt_1}, \eqref{eq:dudt_2} and \eqref{eq:yududy_formula} and using that $\delta^\ast =1$ gives the energy equation 
\[
\frac12 \frac{d}{dt} \|\bs{u}\|^2 + \frac{1}{\Rey} \|\nabla \bs{u}\|^2  = \frac12 C_f(t) (1 - \frac12 \|\bs{u}-\bs{e}_x \|^2). 
\]

\end{proof}

\bibliographystyle{amsplain}
\bibliography{TBL_theory_v.2}

\end{document}